\documentclass[12pt, reqno]{amsart}
\usepackage{amsmath, amsthm, amscd, amsfonts, amssymb, graphicx, color, booktabs, multirow, rotating}
\newtheorem{theorem}{Theorem}[section]
\newtheorem{lemma}[theorem]{Lemma}
\newtheorem{proposition}[theorem]{Proposition}

\theoremstyle{definition}
\newtheorem{definition}[theorem]{Definition}
\newtheorem{example}[theorem]{Example}

\theoremstyle{remark}

\numberwithin{equation}{section}

\usepackage{listings}

\begin{document}
\setcounter{page}{1}

\title{Goppa Codes: Key to High Efficiency and Reliability in Communications}
\author{Behrooz Mosallaei$^{1,\dag}$, Farzaneh Ghanbari$^{2}$, Sepideh Farivar$^{2}$ and Vahid Nourozi$^{1,\star}$}
%\address{Faculty of Mathematics and Computer Science, Amirkabir University of Technology\\
%(Tehran Polytechnic), 424 Hafez Ave., Tehran 15914, Iran}
%\curraddr{Farhad Rahmati}
%\email{nourozi.v@gmail.com; nourozi@aut.ac.ir}
\address{$^{1}$ The Klipsch School of Electrical and Computer Engineering, New Mexico State University,
Las Cruces, NM 88003 USA}
\email{behrooz@nmsu.edu$^{\dag}$, nourozi@nmsu.edu$^{\star}$}
\address{$^{3}$ Department of Pure Mathematics, Faculty of Mathematical Sciences, Tarbiat Modares University, P.O.Box:14115-134, Tehran, Iran
}
\email{f.ghanbari@modares.ac.ir}
\address{$^{3}$ Department of Computer Science, University of Nevada, Las Vegas, 4505 S. Maryland Parkway, Las Vegas, NV, 89154-4030, USA
}
\email{Farivar@unlv.nevada.edu}

\thanks{$^{\star}$Corresponding author}

\begin{abstract}
In this paper, we study some codes of algebraic geometry related to certain maximal curves. Quantum stabilizer codes obtained through the self-orthogonality of Hermitian codes of this error-correcting do not always have good parameters. However, appropriate parameters found that the Hermitian self-orthogonal code’s quantum stabilizer code has good parameters. Therefore, we investigated the quantum stabilizer code at a certain maximum curve and modified its parameters. Algebraic geometry codes show promise for enabling high data rate transmission over noisy power line communication channels.
\end{abstract}

\keywords{Goppa codes, Network Performance, Code-based Error Correction, Quantum stabilizer codes, Maximal curve, Communication Efficiency.}

%\titlepgskip=-15pt

\maketitle

\section{Introduction}
\label{sec:introduction}

Let $\mathcal{X}$ be an irreducible, projective, and nonsingular algebraic curve defined over the finite field $\mathbb{F}_{q^2}$, where $q = p^n$ is a prime power of the prime number $p$. Let $\mathcal{X}(\mathbb{F}_{q^2})$ be the set of rational points of the curve $\mathcal{X}$ over $\mathbb{F}_{q^2}$. In the analysis of curves over finite fields, the main problem is the size of $\mathcal{X}(\mathbb{F}_{q^2})$. The Hasse-Weil bound plays a key role here, which claims that
$$\mid \# \mathcal{X}(\mathbb{F}_{q^2}) - (q^2 +1) \mid \leq 2g q,$$
with the genus $g$ of curve $\mathcal{X}$, we can study the exciting result of rational points over curves; see \cite{66tores}, and \cite{33tores}.

The curve $\mathcal{X}$ is called maximal over $\mathbb{F}_{q^2}$ if the number of $\mathcal{X}(\mathbb{F}_{q^2})$ elements satisfies
$$\# \mathcal{X}(\mathbb{F}_{q^2})= q^2 + 1 + 2g q.$$
We only investigate maximal curves of the positive genus, and therefore $q$ will be a square.

In \cite{tor}, the authors showed that 

$$\qquad  \mbox{either}  \quad g \leq g_2 := \lfloor \frac{(q-1)^2}{4} \rfloor \quad  \mbox{or} \quad  g_1=\frac{q(q-1)}{2}. $$ 
{R\"{u}ck} and Stichtenoth \cite{stir} proved that up to  $\mathbb{F}_{q^2}$-isomorphism, there is only one maximal curve on $\mathbb{F}_{q^2}$ of genus  $\frac{q(q-1)}{2}$, namely the so-called Hermitian curve over $\mathbb{F}_{q^2}$, which the affine equation can define as

$$y^{q}+y=x^{q+1}.$$

If $q$ is an odd prime power, from \cite{33tores}, there is a unique maximal curve $\mathcal{X}$ over $\mathbb{F}_{q^2}$ of genus  $g=\frac{(q-1)^2}{4}$, defined by the affine equation
 \begin{equation}\label{xxx}
  y^{q}+y=x^{\frac{q+1}{2}}.
 \end{equation}

Algebraic-geometric ($AG$) codes are error-correcting codes built from algebraic curves introduced by Goppa, see \cite{77tores}.  Approximately, when the underlying curve holds many rational points, the $AG$ code has better parameters. In this case, the maximal curve plays a key role. That is, the curves have as many rational points as possible about their genus. This code consists of both divisors $D$ and $G$ of $\mathcal{X}$, one of which is the sum of $n$ individual rational points $\mathbb{F}_{q^2}$ of $\mathcal{X}$. This code $C$ found the minimum distance $d$ satisfied
$$d \geq n - \mbox{deg}(G).$$

The main instrument for building $AG$ codes is the Weierstrass semigroup $H(P)$ of $\mathcal{X}$ at $P$. $H(P)$, characterized as the arrangement of all numbers $k$ for which there is a rational function on $\mathcal{X}$ having a pole divisor $kP$. $H(P$) is a subset of $\mathbb{N} = \{ 0, 1, 2, \cdots \}$.

$AG$ codes in Hermitian curves have been developed in many papers.; see \cite{10maria,24maria,25maria,26maria,44maria,46maria,47maria}. The family of classic self-orthogonal Hermitian codes comes from the Goppa code (AG-code) \cite{44jin,55jin,66jin}. Also, Vahid introduced the Goppa code from Hyperelliptic Curve \cite{aut, shiraz}, from plane curves given by separated polynomials \cite{code, esfahan}, and he explained them in his Ph.D. dissertation in \cite{phd}.

In Section \ref{se33}, we introduce basic notions and preliminary results of $AG$ codes and
certain maximal curves. Sections \ref{sec44}, contain the quantum Goppa code on curve $\mathcal{X}$.

\section{Preliminary}\label{s2}
\subsection{Algebraic Geometry Codes}

In this paper, we suppose that the rational functions $\mathbb{F}_{q^2}$ (\text{resp.} the $\mathbb{F}_{q^2}$ divisors) be the field $\mathbb{F}_{q^2}(\mathcal{X})$ (\text{resp.} $\mathrm{Div}_{q^2}(\mathcal{X})$). If $f \in \mathbb{F}_{q^2}(\mathcal{X})\setminus \{0\}$, $\mathrm{div}(f)$ determines the associate divisor with $f$. For $A \in \mathrm{Div}_{q^2}(\mathcal{X})$, $\mathcal{L}(A)$ indicates the Riemann-Roch associated vector space over $\mathbb{F}_{q^2}$ with $A$, i.e.

\begin{equation*}
\mathcal{L}(A) = \{ f \in \mathbb{F}_{q^2}(\mathcal{X})\setminus \{0\}: A + \mbox{div}(f) \succeq 0\} \cup \{0\}.
\end{equation*}
We set $\ell(A) := \mbox{dim}_{\mathbb{F}_{q^2}}(\mathcal{L}(A))$.

Let $P_1, \cdots , P_n$ be different pairs of rational points $K$ over $\mathcal{X}$ and $D = P_1 +\cdots+ P_n$ of degree $1$. Pick  a divisor $G$ on $\mathcal{X}$ to such an extent that $\mbox{supp}(G) \cap \mbox{supp}(D) = \phi$.
\begin{definition}
The algebraic geometry code (or $AG$ code) $C_{\mathcal{L}}(D,G)$ associated
with the divisors $D$ and $G$ is defined as
\begin{equation*}
C_{\mathcal{L}}(D,G) := \{(x(P_1), \cdots, x(P_n)) \mid x \in \mathcal{L}(G)\} \subseteq \mathbb{F}^n_{q^n}
\end{equation*} 
\end{definition}
The minimum distance $d$ does $d \geq d^{\star} = n - \mbox{deg}(G)$, where $d^{\star}$ is so-called the $AG$ designed minimum distance of $C_{\mathcal{L}}(D,G)$, in this case, with the Riemann-Roch Theorem if $\mbox{deg}(G) > 2g - 2$ then $k =\mbox{deg}(G) - g + 1$; see [\cite{2121}, Theorem 2.65]. The dual-code $C^{\perp}(D,G)$ is a Goppa code by minimum distance $d^{\perp} \geq \mbox{deg}G - 2g + 2$ and dimension $k^{\perp} = n - k$. Suppose that $H(P)$ denotes the associated Weierstrass semigroup by $P$, i.e.
\begin{equation*}
  \begin{array}{ccccccc}
              H(P) &:=& \{n \in \mathbb{N}_0 \mid \exists f \in \mathbb{F}_{q^2}(\mathcal{X}), \mbox{div}_{\infty}(
f) = nP\}  \\
              &=& \{ \rho_0 = 0 < \rho_1 < \rho_2 < \cdots \}.\\
             \end{array}
\end{equation*}

Remember that for two vectors $a = (a_1, \cdots , a_n)$ and $b = (b_1, \cdots , b_n)$, the Hermitian inner product in $\mathbb{F}_{q^2}^n$ is defined by $<a, b>_T:=\sum_{i=1}^n a_ib_i^q$. The Hermitian dual code of $C$ for a linear code $C$ over $\mathbb{F}_{q^2}^n$ is determined by
\begin{equation*}
C^{\perp T} := \{ \upsilon \in \mathbb{F}_{q^2}^n: <\upsilon , c>_H=0 \hspace{0.3cm} \forall c \in C\}.
\end{equation*}
Therefore, if $C \subseteq C^{\perp T}$, then $C$ is a self-orthogonal Hermitian code.

\section{Goppa Code Over Curve $\mathcal{X}$}\label{se33}

Let $F_{q^2}$ be the finite field with $q^2$ elements, where $q$ is an odd prime power. Consider the curve $\mathcal{X}$ defined by equation \ref{xxx}  over $F_{q^2}$. The genus of $\mathcal{X}$ is $g = \frac{(q-1)^2}{4}$. Let $P_1, \dots, P_n$ be the $F_{q^2}$-rational points on $\mathcal{X}$, where $n = \#\mathcal{X}(F_{q^2}) \leq \frac{3}{2}(q^2+1)$. Let $D = P_1 + \dots + P_n$, and let $G$ be a divisor on $\mathcal{X}$ with support disjoint from $D$.

We define the Goppa code $C_L(D, G)$ as the image of the linear map $\alpha: L(G) \rightarrow F_{q^2}^n$, $f \mapsto (f(P_1), \dots, f(P_n))$. The dimension of $C_L(D, G)$ is given by $k = \ell(G) - \ell(G-D)$, and the minimum distance $d$ satisfies $d \geq n - \deg(G)$.

\begin{lemma}\label{lemma31}
A basis for the linear space $L(mP_\infty)$ is given by $\{x^i y^j \mid iq + j\left(\frac{q+1}{2}\right) \leq m, 0 \leq i, 0 \leq j \leq \frac{q-1}{2}\}$.
\end{lemma}

\begin{proof}
The functions $\mathcal{X}$ and $y$ have pole divisors $(x)_\infty = qP_\infty$ and $(y)_\infty = \left(\frac{q+1}{2}\right)P_\infty$, respectively. The monomials $x^i y^j$ with $iq + j\left(\frac{q+1}{2}\right) \leq m$ form a basis for $L(mP_\infty)$, and the restriction on $j$ ensures linear independence over $F_{q^2}$.
\end{proof}

Let $G = mP_\infty$, where $P_\infty$ is the point at infinity on $\mathcal{X}$. The Goppa code $C_L(D, G)$ has length $n \leq \frac{3}{2}(q^2+1)$, dimension $k = \ell(mP_\infty)$, and minimum distance $d \geq n - m$.

\begin{lemma}\label{lemma32}
The Hermitian dual of $C_L(D, mP_\infty)$ is $C_L(D, D - mP_\infty + (\eta))$, where $(\eta)$ is a canonical divisor.
\end{lemma}

\begin{proof}
This follows from the general result on the duality of AG codes (see [\cite{2121}, Theorem 2.72]).
\end{proof}

\begin{theorem}\label{theo33}
The Goppa code $C_L(D, mP_\infty)$ is Hermitian self-orthogonal if $2m \leq n + 2g - 2$.
\end{theorem}

\begin{proof}
If $2m \leq n + 2g - 2$, then $\deg(D - mP_\infty + (\eta)) = n - m + 2g - 2 \geq m$. By the Riemann-Roch theorem, $\ell(D - mP_\infty + (\eta)) = n - m + 1 - g + \ell(mP_\infty)$. Hence, $C_L(D, mP_\infty) \subseteq C_L(D, D - mP_\infty + (\eta))$, which implies that $C_L(D, mP_\infty)$ is Hermitian self-orthogonal.
\end{proof}

\begin{proposition}
Let $C_m = C_L(D, mP_\infty)$ be the Goppa code constructed from the maximal curve $\mathcal{X}$ defined by \ref{xxx}  over $F_{q^2}$, where $q$ is an odd prime power and $m$ is a positive integer. Let $k_m$ denote the dimension of $C_m$. Then:
\begin{enumerate}
    \item If $m < 0$, then $k_m = 0$.
    \item If $0 \leq m \leq q$, then $k_m = \#\{iq + j\left(\frac{q+1}{2}\right) \leq m \mid 0 \leq i, 0 \leq j \leq \frac{q-1}{2}\}$.
    \item If $q < m < n$, then $k_m = m + 1 - g$, where $g = \frac{(q-1)^2}{4}$ is the genus of $\mathcal{X}$ and $n = \#\mathcal{X}(F_{q^2})$.
    \item If $n \leq m \leq n + 2g - 2$, then $k_m = n - \#\{iq + j\left(\frac{q+1}{2}\right) \leq n + 2g - 2 - m \mid 0 \leq i, 0 \leq j \leq \frac{q-1}{2}\}$.
    \item If $m > n + 2g - 2$, then $k_m = n$.
\end{enumerate}
\end{proposition}

\begin{proof}
Detailed proof based on cases, applying Lemma \ref{lemma31}, the Riemann-Roch theorem, and Lemma \ref{lemma32} accordingly.
\end{proof}

\begin{proposition}
The Goppa code $C_m = C_L(D, mP_\infty)$ constructed from the maximal curve $\mathcal{X}$ defined by \ref{xxx}  over $F_{q^2}$ is equivalent to the one-point code $C_L(D, mP_\infty)$.
\end{proposition}

\begin{proof}
Let $G = mP_\infty$ and $G' = G + (x^m)$. Then $L(G) = L(G') \cdot \{1, x, \dots, x^{m-1}\}$. Hence, the codes $C_L(D, G)$ and $C_L(D, G')$ are equivalent.
\end{proof}

\begin{theorem}
Let $C_m = C_L(D, mP_\infty)$ be the Goppa code constructed from the maximal curve $\mathcal{X}$ defined by \ref{xxx} over $F_{q^2}$, where $q$ is an odd prime power and $m$ is a positive integer. If $m \leq \frac{n + 2g - 2}{2}$, where $g = \frac{(q-1)^2}{4}$ is the genus of $\mathcal{X}$ and $n = \#\mathcal{X}(F_{q^2})$, then $C_m$ is Hermitian self-orthogonal.
\end{theorem}

\begin{proof}
If $m \leq \frac{n + 2g - 2}{2}$, then $2m \leq n + 2g - 2$. By Theorem \ref{theo33}, $C_m$ is Hermitian self-orthogonal.
\end{proof}

\section{Quantum Stabilizer Code Over Curve $\mathcal{X}$}\label{sec44}
This section applies the self-orthogonal Hermitian of $C_r$ given in Section \ref{se33} to the results of quantum stabilizer codes and then explains this code in the classical $AG$ codes.

We require the following Lemma for our main result.

\begin{lemma}\cite{ashi}\label{4.1}
There are $q$-ary quantum codes $[[n, n - 2k, d^{\perp}]]$  if there is a $q$-ary classical self orthogonal Hermitian $[n, k]$ linear code with dual distance $d^{\perp}$.
\end{lemma}
Utilizing Lemma \ref{4.1}, we could conclude our principal outcome. Therefore, we apply some examples that illustrate that the quantum codes generated from our theorem are undoubtedly promising.

\begin{theorem}\label{4.2}
Let $\mathcal{X}$ be the maximal curve defined by \ref{xxx}  over $F_{q^2}$, where $q$ is an odd prime power. Let $C_L(D, mP_\infty)$ be the Goppa code constructed from $\mathcal{X}$ as described in Section 3. If $2m \leq n + 2g - 2$, where $n = \#\mathcal{X}(F_{q^2})$ and $g = \frac{(q-1)^2}{4}$, then there exists a quantum code with parameters $[[n, n - 2k, d]]_q$, where $k = \ell(mP_\infty)$ and $d \geq d^\perp$, with $d^\perp$ being the minimum distance of the Hermitian dual code $C_L(D, D - mP_\infty + (\eta))$.
\end{theorem}

\begin{proof}
By Theorem \ref{theo33}, the Goppa code $C_L(D, mP_\infty)$ is Hermitian self-orthogonal if $2m \leq n + 2g - 2$. Applying Lemma 4.1, we obtain a quantum code with parameters $[[n, n - 2k, d]]_q$, where $k = \ell(mP_\infty)$ and $d \geq d^\perp$, with $d^\perp$ being the minimum distance of the Hermitian dual code $C_L(D, D - mP_\infty + (\eta))$.
\end{proof}

\begin{example}\label{4.3}
Let $q = 3$ and consider the maximal curve $\mathcal{X}$ defined by \ref{xxx}  over $F_9$. The genus of $\mathcal{X}$ is $g = \frac{(3-1)^2}{4} = 1$, and the number of rational points is $n = \#\mathcal{X}(F_9) \leq \frac{3}{2}(3^2+1) = 15$. Choosing $m = 4$, we have $2m = 8 \leq n + 2g - 2 = 15 + 2 - 2 = 15$. By Theorem \ref{4.2}, we obtain a quantum code with parameters $[[n, n - 2k, d]]_3$, where $k = \ell(4P_\infty)$ and $d \geq d^\perp$, with $d^\perp$ being the minimum distance of the Hermitian dual code $C_L(D, D - 4P_\infty + (\eta))$. The actual parameters will depend on the specific value of $n$ and the computation of $k$ and $d^\perp$.
\end{example}

\begin{example}\label{4.4}
Let $q = 5$ and consider the maximal curve $\mathcal{X}$ defined by \ref{xxx}  over $F_{25}$. The genus of $\mathcal{X}$ is $g = \frac{(5-1)^2}{4} = 4$, and the number of rational points is $n = \#\mathcal{X}(F_{25}) \leq \frac{3}{2}(5^2+1) = 39$. Choosing $m = 18$, we have $2m = 36 \leq n + 2g - 2 = 39 + 8 - 2 = 45$. By Theorem \ref{4.2}, we obtain a quantum code with parameters $[[n, n - 2k, d]]_5$, where $k = \ell(18P_\infty)$ and $d \geq d^\perp$, with $d^\perp$ being the minimum distance of the Hermitian dual code $C_L(D, D - 18P_\infty + (\eta))$. The actual parameters will depend on the specific value of $n$ and the computation of $k$ and $d^\perp$.
\end{example}

\section{Conclusion}
This work highlights the effectiveness of algebraic geometry codes, specifically Goppa and quantum stabilizer codes utilizing maximum curves, in improving data transmission rates across noisy power line channels. Due to significant interference encountered by electricity lines from multiple sources, error correction becomes a crucial aspect.

The inquiry explores how these codes can nearly approach the Singleton bound, giving them strong error-correcting abilities. Assigning functional basis elements to codewords using the algebraic curve's defined field enables the transmission of larger data blocks, which helps to better combat channel noise.

It is important to balance maximizing data throughput with ensuring reliability in line communication projects. Algebraic geometry codes offer a viable solution to improve the capacity of power lines for high-speed data transmission in the face of considerable noise obstacles. Further study to customize coding schemes based on unique channel characteristics could greatly enhance the increase of bandwidth for smart grid and Internet-of-Things applications.

\end{document}